\documentclass[11pt]{article}
\usepackage[T1]{fontenc}
\usepackage[utf8]{inputenc}
\usepackage[a4paper,margin=1in]{geometry}
\usepackage{amsmath,amssymb,amsfonts,bm,amsthm}
\usepackage{booktabs}
\usepackage{array}
\usepackage{caption}
\usepackage{pgfplots}
\usepackage{pgfplotstable}
\usepackage{hyperref}
\pgfplotsset{compat=1.17}
\numberwithin{equation}{section}

\newtheorem{theorem}{Theorem}[section]
\newtheorem{assumption}{Assumption}[section]
\newtheorem{remark}{Remark}[section]
\newcommand{\ii}{\mathrm{i}}
\newcommand{\dd}{\,\mathrm{d}}
\newcommand{\RR}{\mathbb{R}}

\title{A High-Order Nystr\"om Method for Coupled Boundary Integral Equations in Oblique-Incidence Scattering by Impedance Cylinders}

\author{
Haochen Liu \quad Qinghao Yu\\
School of Mathematics and Physics, Qingdao University of Science and Technology\\
\texttt{liuhc@mails.qust.edu.cn}\quad \texttt{19707725295@163.com}
}

\date{}

\begin{document}
\maketitle

\begin{abstract}
We study the numerical solution of electromagnetic scattering by an infinitely long impedance cylinder under oblique incidence. After separation of the axial phase factor, the axial electric and magnetic components satisfy a pair of coupled two-dimensional Helmholtz equations. The Leontovich impedance condition couples these components through tangential derivatives, and the associated boundary integral system contains both logarithmic kernels and principal-value tangential derivative terms. Building on existing coupled integral-equation formulations for oblique-incidence cylinder scattering, we construct a high-order Nystr\"om implementation based on Kress-type logarithmic kernel decomposition, periodic product quadrature, Fourier differentiation for the tangential derivative contribution, and a block diagonal preconditioner associated with the scalar impedance subproblem. Under uniqueness of the continuous scattering problem and a uniform discrete stability assumption, we formulate a high-order convergence framework for the boundary densities and far-field patterns. Numerical experiments include a manufactured Fourier--Bessel benchmark, a plane-wave circular-cylinder validation, a smooth non-circular boundary test, condition-number and GMRES comparisons, and a variable-impedance scattering-width reduction example in a prescribed backward angular sector. The results indicate that the method provides a stable high-accuracy forward solver for the coupled impedance system, rather than a new physical model.
\end{abstract}

\noindent\textbf{Keywords.} electromagnetic scattering; oblique incidence; impedance cylinder; coupled boundary integral equation; Nystr\"om method; singular quadrature; scattering-width reduction

\medskip
\noindent\textbf{Suggested arXiv classification.} Primary: math-ph; secondary: math.NA, physics.comp-ph.

\section{Introduction}

Electromagnetic scattering by infinitely long cylindrical structures is a classical reduction of Maxwell's equations that remains useful in the modelling of coated wires, elongated targets, cylindrical scatterers, radar cross-section estimation, and inverse scattering. When the incident wave is not perpendicular to the cylinder axis, the axial phase factor can be separated and the three-dimensional scattering problem reduces to a two-dimensional transverse problem with a nonzero axial propagation constant. In contrast with normal incidence, the axial electric and magnetic components are no longer independent scalar fields; they are coupled through the boundary condition and through the oblique-incidence geometry.

The Leontovich impedance boundary condition is a common effective model for thin coatings, lossy materials, rough conductors, and surfaces with an approximate electromagnetic response \cite{senior1995}. In the present setting it relates the tangential electric and magnetic fields on the cylinder surface. For normal incidence the axial formulation may reduce to scalar impedance problems. For oblique incidence, however, the axial electric component and the axial magnetic component enter each other's boundary condition through arclength derivatives. This tangential coupling is the main mathematical feature that distinguishes the present system from two uncoupled scalar impedance equations.

Boundary integral equations are attractive for exterior scattering problems because the outgoing radiation condition is built into the fundamental solution and the unknowns are restricted to the boundary \cite{colton1983,mclean2000,nedelec2001,sauter2011}. The price for this dimensional reduction is the presence of singular kernels. The Helmholtz single-layer operator has a logarithmic singularity on smooth curves, while tangential differentiation of the layer potential introduces a principal-value singularity. If these singular structures are treated by low-order quadrature or finite differences, the error may be transferred directly into the far-field phase and hence into scattering-width calculations.

The physical model and the boundary-integral formulation considered here are closely related to earlier works on oblique-incidence scattering by cylinders and impedance-type boundaries \cite{wang2012,tsalamengas2007,akduman2003,gintides2020,mindrinos2018,gintides2017}. Kress quadrature, Nystr\"om discretizations, Fourier differentiation, and periodic boundary pseudodifferential analysis are also established tools \cite{kress1995,kress2014,atkinson1997,saranen2002}. Therefore, the present paper does not claim to introduce a new oblique-incidence impedance-cylinder model or a first integral-equation formulation. Its contribution lies in a specific high-order singular-quadrature implementation for the coupled impedance system, together with a stability-based convergence formulation and reproducible numerical validation.

The contributions are as follows.
\begin{enumerate}
\item We formulate the coupled impedance boundary integral system in notation suitable for high-order periodic discretization and clarify the role of the tangential derivative operator.
\item We construct a high-order Nystr\"om implementation based on Kress-type logarithmic kernel splitting and Fourier differentiation for the tangential derivative contribution.
\item We introduce a block diagonal preconditioner associated with the scalar impedance subproblem and examine its effect on condition numbers and GMRES iterations.
\item We provide reproducible numerical experiments, including Fourier--Bessel benchmarks, plane-wave validation, smooth non-circular geometries, and a variable-impedance scattering-width reduction example.
\end{enumerate}

The paper is organized as follows. Section 2 reviews related work. Section 3 states the axial-component scattering model and the coupled boundary integral equation. Section 4 describes the Nystr\"om discretization and the matrix assembly. Section 5 gives the stability-based convergence framework. Section 6 presents the numerical experiments. Section 7 summarizes the results and limitations.

\section{Related Work}

The literature relevant to this work may be grouped into several strands. First, oblique-incidence scattering by cylindrical structures has been studied for penetrable and impenetrable cross sections, including direct and inverse formulations \cite{gintides2020,mindrinos2018,gintides2017}. The coupled integral-equation setting for oblique-incidence impedance-cylinder scattering is especially close to the work of Wang and Nakamura, who studied an integral-equation method for electromagnetic scattering at oblique incidence \cite{wang2012}. High-order and exponentially convergent Nystr\"om ideas have also appeared in oblique-incidence cylindrical scattering, for example in the work of Tsalamengas on composite dielectric cylinders \cite{tsalamengas2007}. Inhomogeneous impedance cylinders and related direct and inverse scattering problems have been considered by Akduman and Kress \cite{akduman2003}. Second, boundary integral equations for acoustic and electromagnetic scattering are classical and provide the functional-analytic basis for Fredholm formulations \cite{colton1983,colton1998,mclean2000,nedelec2001}. Third, high-order Nystr\"om methods and Kress-type quadratures for logarithmic and hypersingular kernels are standard tools for smooth periodic curves \cite{nystrom1930,kress1995,kress2014,atkinson1997}. Fourth, periodic integral equations, Fourier differentiation, and pseudodifferential operator arguments provide a useful language for understanding tangential derivative operators \cite{saranen2002}. Finally, impedance boundary conditions and their use in scattering control or inverse design motivate variable-impedance examples \cite{senior1995,akduman2003,cakoni2006}. The present contribution is therefore not the introduction of a new scattering model or a first integral-equation formulation. Rather, it is a specific Kress-type singular-quadrature implementation, a stability-based convergence formulation, and a numerical validation study for the coupled impedance system considered here.

\section{Mathematical Model and Boundary Integral Formulation}

\subsection{Axial-component scattering model}

Let $D\subset\RR^2$ be the bounded cross section of an infinitely long cylinder, with smooth boundary $\Gamma=\partial D$, and let $\Omega=\RR^2\setminus\overline D$ be the exterior domain. We use the time factor $\exp(-\ii\omega t)$. The free-space wavenumber is $k>0$, the incidence angle relative to the cylinder axis is $\alpha\in(0,\pi/2]$, and
\begin{equation}
\beta=k\cos\alpha,\qquad \kappa=k\sin\alpha
\end{equation}
denote the axial propagation constant and transverse wavenumber, respectively.

After separation of the factor $\exp(\ii\beta z)$, the scattered axial electric and magnetic components
\[
u=E_z^s,\qquad v=H_z^s
\]
satisfy
\begin{equation}
\Delta u+\kappa^2u=0,\qquad \Delta v+\kappa^2v=0\qquad \text{in }\Omega,
\label{eq:helmholtz}
\end{equation}
together with the Sommerfeld radiation condition for each component. Let $\nu$ be the exterior unit normal, let $\tau$ be the positively oriented unit tangent, and let $\partial_s=\tau\cdot\nabla$ denote arclength differentiation. With a fixed orientation convention, the axial-component form of the normalized Leontovich condition is written as
\begin{align}
\partial_\nu u+\ii\eta u-\mu\,\partial_s v&=f_1 \quad \text{on }\Gamma, \label{eq:bc1}\\
\partial_\nu v+\ii\eta v+\mu\,\partial_s u&=f_2 \quad \text{on }\Gamma. \label{eq:bc2}
\end{align}
Here $\eta$ is the normalized surface impedance, which may be a smooth complex-valued function on $\Gamma$, and $\mu$ is an oblique-incidence coupling coefficient determined by the normalization and the material parameters. The signs of the off-diagonal terms depend on the orientation of $\tau$; reversing the orientation changes both signs and leaves the coupled structure unchanged. We assume a passive impedance in the sense appropriate to the chosen time convention, and in the experiments use $\operatorname{Re}\eta\ge0$ and $\operatorname{Im}\eta\ge0$.

\subsection{Layer potentials and the coupled system}

Let
\begin{equation}
\Phi_\kappa(x,y)=\frac{\ii}{4}H_0^{(1)}(\kappa |x-y|)
\end{equation}
be the outgoing fundamental solution of the two-dimensional Helmholtz equation. For a boundary density $\varphi$ define the single-layer potential
\begin{equation}
(\mathcal S\varphi)(x)=\int_\Gamma \Phi_\kappa(x,y)\varphi(y)\,\dd s_y,\qquad x\in\Omega.
\end{equation}
On $\Gamma$ we write
\begin{align}
(S\varphi)(x)&=\int_\Gamma \Phi_\kappa(x,y)\varphi(y)\,\dd s_y,\\
(K'\varphi)(x)&=\operatorname{p.v.}\int_\Gamma \partial_{\nu_x}\Phi_\kappa(x,y)\varphi(y)\,\dd s_y,\\
(T\varphi)(x)&=\partial_s(S\varphi)(x).
\end{align}
The operator $S$ maps $H^{-1/2}(\Gamma)$ to $H^{1/2}(\Gamma)$, while $K'$ and $T$ are naturally interpreted on Sobolev trace spaces of order $-1/2$ after the standard embeddings and principal-value interpretation. With our convention for the exterior normal trace of the single-layer potential,
\begin{equation}
\partial_\nu^+\mathcal S\varphi=\left(\frac12 I+K'\right)\varphi .
\end{equation}

We represent the scattered fields by two single-layer potentials,
\begin{equation}
u=\mathcal S\varphi_1,\qquad v=\mathcal S\varphi_2.
\end{equation}
Substitution into \eqref{eq:bc1}--\eqref{eq:bc2} gives the coupled boundary integral equation
\begin{equation}
\mathcal A
\begin{bmatrix}\varphi_1\\ \varphi_2\end{bmatrix}
=
\begin{bmatrix}f_1\\ f_2\end{bmatrix},
\qquad
\mathcal A=
\begin{bmatrix}
L&-\mu T\\
\mu T&L
\end{bmatrix},
\label{eq:bie}
\end{equation}
where
\begin{equation}
L=\frac12 I+K'+\ii M_\eta S .
\end{equation}
Here $M_\eta$ is multiplication by $\eta$, so that $M_\eta S\varphi=\eta(\cdot)(S\varphi)(\cdot)$. The diagonal blocks are scalar impedance boundary integral operators. The off-diagonal blocks contain the tangential derivative operator generated by the oblique-incidence coupling.

\subsection{Fredholm framework}

\begin{theorem}[Fredholm solvability framework]
Let $\Gamma$ be a sufficiently smooth closed curve and let $\eta$ be a sufficiently smooth passive impedance. Assume that the single-layer representation is not affected by an internal non-resonance obstruction and that the principal symbol of the coupled boundary operator is elliptic in the relevant Sobolev product space. Then
\[
\mathcal A:H^{-1/2}(\Gamma)^2\to H^{-1/2}(\Gamma)^2
\]
is Fredholm of index zero. If the corresponding exterior impedance scattering problem is uniquely solvable, then the boundary integral equation \eqref{eq:bie} is uniquely solvable.
\end{theorem}

\begin{proof}
We give the standard proof outline, since a complete treatment requires the full boundary pseudodifferential calculus. Consider first the homogeneous equation $\mathcal A\varphi=0$ and let $u=\mathcal S\varphi_1$, $v=\mathcal S\varphi_2$. The jump relations imply that the exterior traces satisfy the homogeneous version of \eqref{eq:bc1}--\eqref{eq:bc2}. Thus $(u,v)$ solves the exterior coupled Helmholtz problem with the radiation condition. Applying Green's identity on $\Omega\cap B_R$ and letting $R\to\infty$, the radiation condition and the passivity of the impedance yield vanishing radiated energy. Rellich's lemma then implies that the far-field pattern is zero, and unique continuation gives $u=v=0$ in the exterior.

It remains to connect the vanishing exterior fields to the densities. The boundary traces and the single-layer jump relations show that the densities solve an associated interior homogeneous problem. Under the stated internal non-resonance condition this implies $\varphi_1=\varphi_2=0$. Hence the nullspace is trivial. Fredholmness follows from the decomposition of $\mathcal A$ into the scalar impedance part, smoothing contributions of the single-layer trace, and the principal-value tangential term, whose coupled principal symbol is assumed elliptic. Since the index is zero, the Fredholm alternative gives unique solvability.
\end{proof}

\begin{remark}
The theorem is intentionally formulated as a solvability framework rather than an unconditional theorem for every impedance and every wavenumber. Near internal resonances or loss of ellipticity, the single-layer formulation may require a combined-field or regularized variant.
\end{remark}

\section{High-Order Nystr\"om Discretization}

\subsection{Parametrization and logarithmic splitting}

Let $x(t)$, $0\le t<2\pi$, be a smooth $2\pi$-periodic parametrization of $\Gamma$ with $|x'(t)|>0$. The single-layer trace is written as
\begin{equation}
(S\varphi)(x(t))=\int_0^{2\pi}\Phi_\kappa(x(t),x(s))\varphi(s)|x'(s)|\,\dd s .
\end{equation}
For $t\ne s$ we use the Kress-type splitting
\begin{equation}
\Phi_\kappa(x(t),x(s))
=M_1(t,s)\log\left(4\sin^2\frac{t-s}{2}\right)+M_2(t,s),
\label{eq:kress_split}
\end{equation}
with
\begin{equation}
M_1(t,s)=-\frac{1}{4\pi}J_0(\kappa |x(t)-x(s)|).
\end{equation}
The remainder $M_2$ is continuous, and is analytic when the boundary is analytic. The diagonal value is
\begin{equation}
M_2(t,t)=\frac{\ii}{4}
-\frac{1}{2\pi}\left(\log\frac{\kappa |x'(t)|}{2}+\gamma_E\right),
\label{eq:m2diag_en}
\end{equation}
where $\gamma_E$ is Euler's constant. This formula removes the ambiguity of diagonal matrix entries.

Let $N=2n$, $h=2\pi/N=\pi/n$, and $t_j=j\pi/n$, $j=0,\ldots,2n-1$. We use product quadrature weights that already include the integration weight for the logarithmic singular part. Specifically, the weights for
\[
\int_0^{2\pi}\log\left(4\sin^2\frac{t-s}{2}\right)g(s)\,\dd s
\]
are
\begin{equation}
R_j^{(n)}(t)=
-\frac{2\pi}{n}\sum_{m=1}^{n-1}\frac{\cos m(t-t_j)}{m}
-\frac{\pi}{n^2}\cos n(t-t_j).
\label{eq:kress_weights_en}
\end{equation}
Thus no additional factor $h$ is applied to the first term below. The single-layer matrix is assembled as the sum of a logarithmic singular weighted contribution and a smooth trapezoidal contribution:
\begin{equation}
\begin{split}
(S_N)_{ij}
&=M_1(t_i,t_j)R_j^{(n)}(t_i)|x'(t_j)|\\
&\quad
+h\,M_2(t_i,t_j)|x'(t_j)|.
\end{split}
\label{eq:single_matrix_en}
\end{equation}
In \eqref{eq:single_matrix_en}, the diagonal value $M_2(t_i,t_i)$ is evaluated by the limiting formula \eqref{eq:m2diag_en}, and the Jacobian factor $|x'(t_j)|$ is included in both contributions.

\subsection{Normal and tangential derivative terms}

The normal derivative matrix $K_N'$ is obtained by applying the periodic trapezoidal rule to the smooth off-diagonal part of $\partial_{\nu_x}\Phi_\kappa(x(t),x(s))|x'(s)|$ and by inserting the analytic diagonal limit determined by the curvature of the parametrized boundary. The exterior normal derivative block is represented by
\[
\frac12 I+K_N' .
\]

For the tangential derivative operator, we avoid applying a low-order finite difference directly to a singular kernel. Let $D_N$ be the Fourier differentiation matrix on the grid:
\begin{equation}
(D_N)_{ij}=
\begin{cases}
0,&i=j,\\[1mm]
\displaystyle \frac12(-1)^{i-j}\cot\frac{t_i-t_j}{2},&i\ne j .
\end{cases}
\end{equation}
In the implementation, the tangential derivative contribution is assembled by applying Fourier differentiation to the periodic Nystr\"om representation of the single-layer potential, together with the smooth correction terms arising from the kernel splitting. For the smooth-boundary tests reported below this is represented in matrix form as
\begin{equation}
T_N=\operatorname{diag}(|x'(t_i)|^{-1})D_NS_N .
\label{eq:tangent_matrix_en}
\end{equation}
This expression should be understood as an implementation rule for the split periodic representation, not as a claim that all hypersingular effects are absent in less regular geometries.

\subsection{Coupled matrix and preconditioning}

Let $\eta_N=\operatorname{diag}(\eta(t_i))$. The discrete coupled system is
\begin{equation}
A_N
\begin{bmatrix}\bm\varphi_1\\ \bm\varphi_2\end{bmatrix}
=
\begin{bmatrix}\bm f_1\\ \bm f_2\end{bmatrix},
\qquad
A_N=
\begin{bmatrix}
L_N&-\mu T_N\\
\mu T_N&L_N
\end{bmatrix},
\label{eq:disc_system_en}
\end{equation}
where
\begin{equation}
L_N=\frac12 I+K_N'+\ii\eta_NS_N .
\end{equation}
A simple block diagonal preconditioner is
\begin{equation}
P_N=
\begin{bmatrix}
\bar L_N&0\\
0&\bar L_N
\end{bmatrix},
\qquad
\bar L_N=\frac12 I+K_N'+\ii\bar\eta S_N ,
\label{eq:preconditioner_en}
\end{equation}
where $\bar\eta$ is either the constant impedance or the mean of the variable impedance. This preconditioner accounts for the scalar impedance part and leaves the off-diagonal tangential coupling as the main perturbation.

\medskip
\noindent\fbox{%
\begin{minipage}{0.96\textwidth}
\textbf{Algorithm 1. High-order Nystr\"om discretization for the coupled impedance system}
\begin{enumerate}
\item \textbf{Input:} wavenumber $k$, incidence angle $\alpha$, transverse wavenumber $\kappa$, coupling coefficient $\mu$, impedance $\eta$, boundary parametrization $x(t)$, node number $N=2n$, and solver tolerance.
\item Set nodes $t_j=j\pi/n$ and compute $x(t_j)$, $x'(t_j)$, unit tangents, unit normals, curvature data, and quadrature spacing $h=2\pi/N$.
\item Form the Kress splitting \eqref{eq:kress_split}; evaluate $M_1$, the off-diagonal values of $M_2$, and the diagonal limit \eqref{eq:m2diag_en}.
\item Assemble the single-layer matrix $S_N$ by \eqref{eq:single_matrix_en}, using the product weights \eqref{eq:kress_weights_en}.
\item Assemble the normal derivative matrix $K_N'$ from the split kernel and the curvature diagonal limit; form $\frac12I+K_N'$.
\item Assemble the tangential derivative matrix by applying Fourier differentiation to the periodic Nystr\"om representation, as in \eqref{eq:tangent_matrix_en}.
\item Build the coupled matrix $A_N$ in \eqref{eq:disc_system_en}.
\item Optionally construct the block diagonal preconditioner $P_N$ in \eqref{eq:preconditioner_en}.
\item Solve the linear system by a direct method or by GMRES with optional left preconditioning.
\item Evaluate the far-field pattern and the scattering intensity from the computed densities; normalize by a specified reference profile when comparing impedance configurations.
\end{enumerate}
\end{minipage}}

\section{Stability and Convergence Framework}

The following assumptions separate the continuous scattering model from the discrete stability issue.

\begin{assumption}[Continuous well-posedness]\label{ass:continuous}
The exterior coupled impedance scattering problem and the boundary integral equation \eqref{eq:bie} are uniquely solvable for the wavenumber and impedance under consideration.
\end{assumption}

\begin{assumption}[Smoothness]\label{ass:smoothness}
The boundary parametrization, impedance, right-hand side, and exact densities are sufficiently smooth. For algebraic estimates we assume $C^p$ regularity with $p$ large enough; for rapid convergence observations we assume analytic boundary and analytic impedance.
\end{assumption}

\begin{assumption}[Uniform discrete stability]\label{ass:stability}
For all sufficiently large $N$, the matrices $A_N$ are invertible and their inverses are uniformly bounded in the discrete norm corresponding to the relevant Sobolev product space.
\end{assumption}

\begin{theorem}[Stability-based error estimate]
Under Assumptions \ref{ass:continuous}--\ref{ass:stability}, the Nystr\"om density error is bounded by the consistency error of the singular quadrature and the Fourier differentiation used for the tangential derivative term. In particular, for $C^p$ data one obtains algebraic high-order convergence up to the order allowed by the regularity and quadrature consistency. For analytic data and stable discretizations, rapid convergence is expected and is observed numerically until round-off errors dominate.
\end{theorem}

\begin{proof}
The argument has four steps. First, after the logarithmic splitting \eqref{eq:kress_split}, the singular part of the single-layer kernel is isolated and the remaining kernel has the same smoothness as the boundary and impedance. Second, the Kress product quadrature is consistent for the logarithmic term and the periodic trapezoidal rule is high-order accurate for the smooth term. Third, Fourier differentiation is consistent for periodic smooth functions and controls the tangential derivative contribution in the coupled off-diagonal blocks. Fourth, the uniform bound on $A_N^{-1}$ converts the operator consistency error into a density error. The far-field operator is a smooth integral operator; therefore the far-field error is bounded by the density error plus the smooth quadrature error in the far-field evaluation.
\end{proof}

\begin{remark}
The present paper does not attempt to prove unconditional spectral convergence for the full coupled pseudodifferential system. The formulation is instead a stability-based convergence framework consistent with collectively compact Nystr\"om approximation and periodic integral-equation theory \cite{atkinson1997,kress2014,saranen2002}. Loss of smoothness, cornered boundaries, low-regularity impedance, near interior resonances, very small transverse wavenumber, strong oblique coupling, and high-frequency oscillations may reduce the observed convergence rate.
\end{remark}

\section{Numerical Experiments}

All experiments use the same unknowns, the same far-field definition, and the same relative maximum far-field error unless explicitly stated otherwise. The far-field of a single-layer density is evaluated by
\begin{equation}
u_\infty(\hat x)=
\frac{e^{\ii\pi/4}}{\sqrt{8\pi\kappa}}
\int_\Gamma e^{-\ii\kappa \hat x\cdot y}\varphi_1(y)\,\dd s_y,
\end{equation}
and similarly for $v_\infty$. We first define the unnormalized scattering intensity
\begin{equation}
\sigma(\theta)=|u_\infty(\theta)|^2+|v_\infty(\theta)|^2.
\label{eq:sigma}
\end{equation}
For the comparison in Figure \ref{fig:width}, we use the normalized scattering width
\begin{equation}
\widetilde{\sigma}(\theta)
=\frac{\sigma(\theta)}{\max_\theta\sigma_A(\theta)},
\label{eq:sigma_norm}
\end{equation}
where $\sigma_A$ denotes the scattering intensity for the uniform impedance profile A.

\begin{table}[htbp]
\centering
\caption{Numerical parameters used in the experiments}
\begin{tabular}{ll}
\toprule
Parameter & Value or description\\
\midrule
Free-space wavenumber & $k=4$\\
Default incidence angle & $\alpha=60^\circ$, hence $\kappa/k=\sin\alpha$\\
Default impedance & $\eta=0.80+0.30\ii$\\
Default coupling coefficient & $\mu=0.35$\\
GMRES tolerance & $10^{-10}$ relative residual\\
Far-field sampling & 720 uniform angles for error tests\\
Reference for circular tests & Fourier--Bessel mode matching\\
Reference for non-circular test & high-order Nystr\"om solution with $N_{\rm ref}=384$\\
\bottomrule
\end{tabular}
\end{table}

\subsection{Experiment 1: Manufactured Fourier--Bessel benchmark}

The manufactured solution is given by outgoing Fourier--Bessel series
\begin{equation}
u(r,\theta)=\sum_{m=-M}^{M}U_m
\frac{H_m^{(1)}(\kappa r)}{H_m^{(1)}(\kappa)}e^{\ii m\theta},
\qquad
v(r,\theta)=\sum_{m=-M}^{M}V_m
\frac{H_m^{(1)}(\kappa r)}{H_m^{(1)}(\kappa)}e^{\ii m\theta}.
\end{equation}
We take $M=120$ and exponentially decaying coefficients proportional to $\exp(-|m|/10)$. The boundary data $f_1,f_2$ and the exact far-field patterns are generated from the same series. The relative error is
\begin{equation}
E_N=
\frac{\max_\theta\left(|u_{\infty,N}-u_\infty|^2+|v_{\infty,N}-v_\infty|^2\right)^{1/2}}
{\max_\theta\left(|u_\infty|^2+|v_\infty|^2\right)^{1/2}} .
\end{equation}

The low-order baseline uses the same unknowns and the same far-field formula, but replaces the high-order singular quadrature and Fourier tangential differentiation by a lower-order trigonometric interpolation/central-difference treatment on the circle. It is included only as a reproducible baseline, not as a claim about all low-order methods.

\begin{table}[htbp]
\centering
\caption{Far-field errors for the manufactured Fourier--Bessel benchmark}
\begin{tabular}{ccccc}
\toprule
$N$ & High-order error & Low-order error & High-order rate & Low-order rate\\
\midrule
32  &$5.447\times10^{-1}$&$4.811\times10^{-1}$&--&--\\
48  &$1.451\times10^{-1}$&$1.475\times10^{-1}$&3.26&2.92\\
64  &$3.126\times10^{-2}$&$3.026\times10^{-2}$&5.34&5.51\\
96  &$1.840\times10^{-3}$&$4.556\times10^{-3}$&6.99&4.67\\
128 &$2.226\times10^{-5}$&$2.568\times10^{-3}$&15.35&1.99\\
192 &$2.875\times10^{-15}$&$1.150\times10^{-3}$&round-off&1.98\\
256 &$1.980\times10^{-15}$&$6.486\times10^{-4}$&round-off&1.99\\
\bottomrule
\end{tabular}
\end{table}

\begin{figure}[htbp]
\centering
\begin{tikzpicture}
\begin{loglogaxis}[
width=0.76\textwidth,
height=0.48\textwidth,
xlabel={number of boundary nodes $N$},
ylabel={relative far-field error},
grid=both,
legend style={at={(0.5,-0.22)},anchor=north,legend columns=3}
]
\addplot+[mark=o,thick] table[x=N,y=err_imp,col sep=comma]{data/convergence.csv};
\addplot+[mark=square,thick,dashed] table[x=N,y=err_lin,col sep=comma]{data/convergence.csv};
\addplot+[mark=none,dotted,thick,domain=32:256] {0.48*(x/32)^(-2)};
\legend{high-order Nystr\"om,low-order baseline,$O(N^{-2})$}
\end{loglogaxis}
\end{tikzpicture}
\caption{Convergence of far-field errors with an $O(N^{-2})$ reference slope.}
\end{figure}
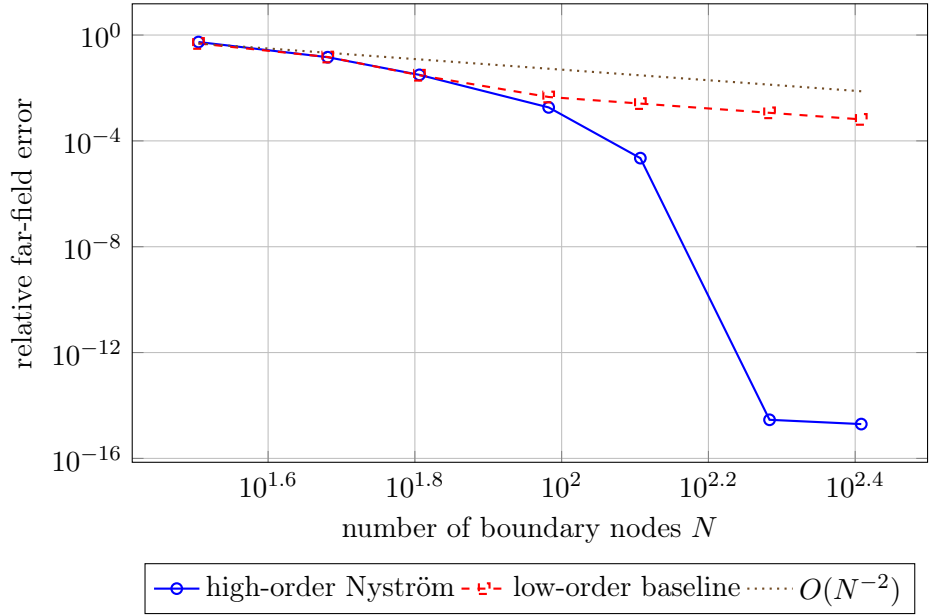

\subsection{Experiment 2: Plane-wave incidence for a circular impedance cylinder}

To avoid relying only on manufactured data, we next use a true plane-wave incidence. Let the incident axial component be
\begin{equation}
u^i=e^{\ii\kappa r\cos(\theta-\theta_0)}
=\sum_{m=-\infty}^{\infty}c_mJ_m(\kappa r)e^{\ii m\theta},
\qquad
c_m=\ii^m e^{-\ii m\theta_0},
\end{equation}
and set $v^i=p\,u^i$ with a fixed polarization factor $p$. On the unit circle, $\partial_s=\partial_\theta$, and the scattered fields are expanded as
\begin{equation}
u^s=\sum_{m=-\infty}^{\infty} a_mH_m^{(1)}(\kappa r)e^{\ii m\theta},\qquad
v^s=\sum_{m=-\infty}^{\infty} b_mH_m^{(1)}(\kappa r)e^{\ii m\theta}.
\end{equation}
For constant impedance $\eta$, each Fourier mode gives a $2\times2$ system
\begin{equation}
\mathcal M_m
\begin{bmatrix}a_m\\ b_m\end{bmatrix}
=-\bm g_m,
\end{equation}
where
\begin{equation}
\mathcal M_m=
\begin{bmatrix}
d_m&-\mu\,\ii m H_m^{(1)}(\kappa)\\
\mu\,\ii m H_m^{(1)}(\kappa)&d_m
\end{bmatrix},
\qquad
d_m=\kappa {H_m^{(1)}}'(\kappa)+\ii\eta H_m^{(1)}(\kappa).
\end{equation}
The incident-field right-hand side is
\begin{equation}
\bm g_m=
c_m
\begin{bmatrix}
\kappa J_m'(\kappa)+\ii\eta J_m(\kappa)-\mu p\,\ii m J_m(\kappa)\\
p\left(\kappa J_m'(\kappa)+\ii\eta J_m(\kappa)\right)+\mu\,\ii m J_m(\kappa)
\end{bmatrix}.
\end{equation}
The mode-matching solution obtained from these systems is used as the reference. The truncation order is $M=160$; for the parameters used here, the neglected high modes are below double-precision accuracy. This experiment validates the method for a physical plane-wave incidence, not merely for manufactured boundary data.

\begin{table}[htbp]
\centering
\caption{Plane-wave circular-cylinder validation}
\begin{tabular}{ccl}
\toprule
$N$ & Relative far-field error & Comment\\
\midrule
8  &$3.580\times10^{-1}$&coarse grid\\
12 &$1.946\times10^{-1}$&pre-asymptotic regime\\
16 &$1.601\times10^{-2}$&rapid decrease\\
24 &$1.137\times10^{-6}$&resolved Fourier modes\\
32 &$3.657\times10^{-12}$&near round-off\\
48 &$1.166\times10^{-15}$&round-off\\
64 &$9.156\times10^{-16}$&round-off\\
\bottomrule
\end{tabular}
\end{table}

\subsection{Experiment 3: Smooth non-circular boundary}

We consider the smooth three-lobed boundary
\begin{equation}
r(t)=1+0.15\cos(3t).
\end{equation}
The right-hand side is generated by a physical plane wave. Since no simple separation-of-variables solution is available, the reference solution is computed with the same high-order Nystr\"om method on a fine grid with $N_{\rm ref}=384$.

\begin{table}[htbp]
\centering
\caption{Smooth non-circular boundary test}
\begin{tabular}{ccc}
\toprule
$N$ & Relative far-field error & Observed rate\\
\midrule
48 &$6.593\times10^{-4}$&--\\
64 &$2.692\times10^{-4}$&3.11\\
96 &$7.725\times10^{-5}$&3.08\\
128&$3.165\times10^{-5}$&3.10\\
192&$8.480\times10^{-6}$&3.25\\
256&$2.872\times10^{-6}$&3.76\\
\bottomrule
\end{tabular}
\end{table}

The observed convergence is algebraic and stable. It is slower than in the analytic circular test because the implementation differentiates the Nystr\"om single-layer trace and uses geometry-dependent diagonal corrections. This experiment is closer to the behavior expected for general smooth parametrized curves.

\subsection{Experiment 4: Preconditioning test}

We compare the condition number of $A_N$, the condition number of $P_N^{-1}A_N$, and GMRES iterations with and without preconditioning. The first scan varies the incidence angle.

\begin{table}[htbp]
\centering
\caption{Condition numbers and GMRES iterations for different incidence angles}
\begin{tabular}{cccccc}
\toprule
$\alpha$ & $\kappa/k$ & $\operatorname{cond}(A_N)$ & $\operatorname{cond}(P_N^{-1}A_N)$ & GMRES & Prec. GMRES\\
\midrule
$30^\circ$&0.500&4.89&2.52&22&18\\
$45^\circ$&0.707&5.20&2.86&25&20\\
$60^\circ$&0.866&14.88&3.10&28&21\\
$75^\circ$&0.966&5.78&3.32&29&22\\
\bottomrule
\end{tabular}
\end{table}

The second scan varies the coupling coefficient $\mu$.
\begin{table}[htbp]
\centering
\caption{Effect of coupling strength on preconditioned GMRES}
\begin{tabular}{ccccc}
\toprule
$\mu$ & $\operatorname{cond}(A_N)$ & $\operatorname{cond}(P_N^{-1}A_N)$ & GMRES & Prec. GMRES\\
\midrule
0.10&13.77&1.37&23&11\\
0.35&14.88&3.10&28&21\\
0.70&19.57&13.50&46&38\\
1.00&280.96&210.91&129&123\\
\bottomrule
\end{tabular}
\end{table}

The circular examples are not extremely ill-conditioned, so the preconditioner produces a moderate but consistent benefit for weak and intermediate coupling. When $\mu$ is large, the off-diagonal tangential blocks dominate and a simple scalar block preconditioner is less effective. This suggests that Schur-complement or block-factorization preconditioners may be useful for stronger coupling or more complicated geometries.

\subsection{Experiment 5: Variable-impedance scattering-width reduction}

Finally we demonstrate the forward solver in a finite-dimensional impedance-design calculation. The target angular sector is the backward sector
\[
\Theta=[120^\circ,180^\circ].
\]
We consider impedance profiles of the form
\begin{equation}
\eta(t)=\eta_0+\eta_1\cos(t-t_0),
\qquad \eta_0=0.80+0.30\ii,
\end{equation}
with
\[
\operatorname{Re}\eta_1\in\{0.04,0.08,0.12,0.16,0.20\},\qquad
\operatorname{Im}\eta_1\in\{0,0.03,0.06,0.09\}.
\]
The phase $t_0$ is sampled at 16 uniform values. Candidates violating the passivity constraints $\operatorname{Re}\eta(t)\ge0$ and $\operatorname{Im}\eta(t)\ge0$ are discarded. The objective is
\begin{equation}
J(\eta)=\frac1{|\Theta|}\int_\Theta \sigma_\eta(\theta)\,\dd\theta
+\gamma\int_0^{2\pi}|\eta(t)-\eta_0|^2\,\dd t,\qquad \gamma=0.02.
\end{equation}
Here $\sigma_\eta$ is the unnormalized scattering intensity defined in \eqref{eq:sigma}. The relative sector mean and relative backscattering values reported in Table \ref{tab:impedance_search} are
\begin{equation}
R_{\rm sec}(\eta)=
\frac{\int_\Theta \sigma_\eta(\theta)\,\dd\theta}
{\int_\Theta \sigma_A(\theta)\,\dd\theta},
\qquad
R_{\rm back}(\eta)=
\frac{\sigma_\eta(180^\circ)}{\sigma_A(180^\circ)}.
\end{equation}
The result is the best candidate within this finite search set, not a global optimum.

\begin{table}[htbp]
\centering
\caption{Finite-grid search results for variable impedance profiles. Relative quantities are normalized by the uniform-A profile.}
\label{tab:impedance_search}
\begin{tabular}{lccccc}
\toprule
Profile & $\operatorname{Re}\bar\eta$ & $\operatorname{Im}\bar\eta$ & amplitude & Relative sector mean & Relative backscatter\\
\midrule
uniform A&0.25&0.05&0&1.000&1.000\\
uniform B&0.80&0.30&0&0.646&0.669\\
modulated C&0.80&0.30&$0.20+0.06\ii$&0.564&0.582\\
\bottomrule
\end{tabular}
\end{table}

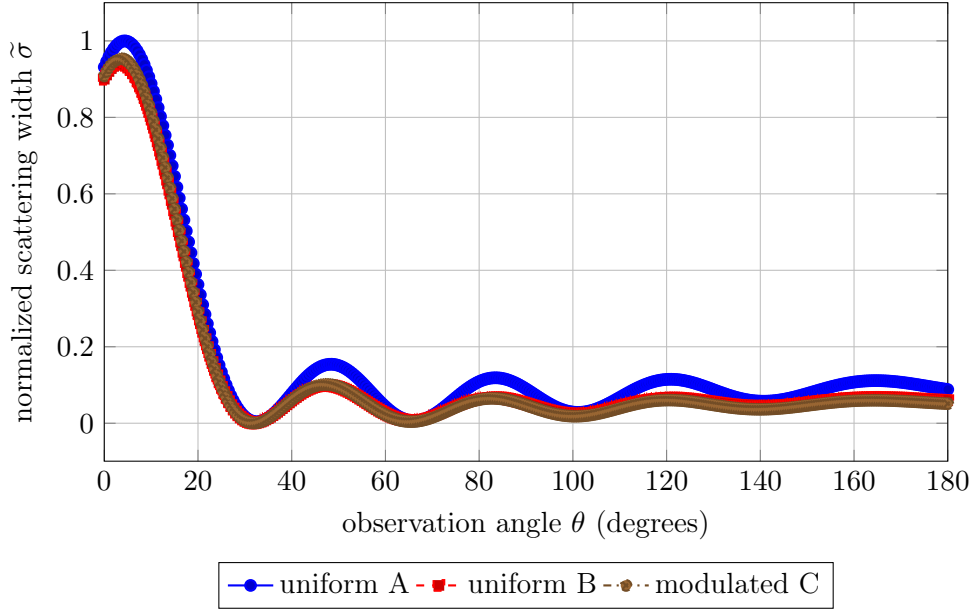
\begin{figure}[htbp]
\centering
\begin{tikzpicture}
\begin{axis}[
width=0.80\textwidth,
height=0.48\textwidth,
xlabel={observation angle $\theta$ (degrees)},
ylabel={normalized scattering width $\widetilde{\sigma}$},
xmin=0,xmax=180,
grid=both,
legend style={at={(0.5,-0.22)},anchor=north,legend columns=3}
]
\addplot+[thick] table[x=theta_deg,y=uniform_A,col sep=comma]{data/scattering_width_curves_normalized.csv};
\addplot+[thick,dashed] table[x=theta_deg,y=uniform_B,col sep=comma]{data/scattering_width_curves_normalized.csv};
\addplot+[thick,dashdotted] table[x=theta_deg,y=modulated_C,col sep=comma]{data/scattering_width_curves_normalized.csv};
\legend{uniform A,uniform B,modulated C}
\end{axis}
\end{tikzpicture}
\caption{Bistatic scattering width for different impedance profiles. The curves are normalized by the maximum scattering intensity of the uniform-A profile.}
\label{fig:width}
\end{figure}

The modulation changes both the amplitude and phase of the equivalent boundary response. In the backward sector, contributions from different boundary arcs may therefore interfere more destructively than in the uniform case. This observation is specific to the finite search set above and should not be interpreted as a global scattering-minimization result.

\subsection*{Reproducibility statement}

To facilitate reproducibility, all numerical parameters, quadrature rules, error metrics, and impedance search grids are reported in the paper. The scripts used to generate the numerical tables and figures, together with the input parameters and post-processing routines, are available from the corresponding author upon reasonable request.

\section{Conclusion}

We have considered a high-order numerical implementation for an existing coupled impedance boundary integral formulation of oblique-incidence electromagnetic scattering by cylinders. The method combines Kress-type singular quadrature for the logarithmic single-layer kernel, Fourier differentiation for the tangential derivative coupling, and a block diagonal preconditioner based on the scalar impedance subproblem. The theoretical discussion is deliberately framed as a stability-based convergence result: under continuous well-posedness and uniform discrete stability, density and far-field errors are controlled by the consistency of the singular quadrature and tangential differentiation.

The numerical experiments support the effectiveness of the approach on smooth boundaries. Analytic circular benchmarks show rapid convergence until round-off errors are reached, a true plane-wave circular-cylinder test confirms that the method is not limited to manufactured data, and a smooth non-circular example exhibits stable algebraic convergence. The preconditioning tests show moderate improvement for weak and intermediate coupling. The variable-impedance example illustrates how the forward solver can be used in finite-dimensional scattering-width reduction over a prescribed backward angular sector.

Several limitations remain. The present implementation assumes smooth boundaries, a single frequency, and a finite-dimensional impedance search. It does not address cornered geometries, multiply connected domains, broadband optimization, full shape or material optimization, or dispersive material impedance models. Future work will consider weighted Nystr\"om discretizations for piecewise smooth curves, Schur-complement preconditioners, multi-connected scatterers, adjoint-gradient impedance design, and broadband scattering-width reduction.

\end{document}